\documentclass{article}

\usepackage[utf8]{inputenc}
\usepackage{float}
\usepackage{cite}
\usepackage{amsthm}
\usepackage[all]{xy}
\usepackage{amsmath}
\usepackage{amssymb}
\usepackage{amsfonts}
\usepackage{stmaryrd}

\usepackage{enumerate}

\newtheorem{theo}{Theorem}

\newtheorem{prop}{Proposition}
\newtheorem{lemma}{Lemma}

\theoremstyle{definition}

\theoremstyle{definition}
\newtheorem{defi}{Definition}
\theoremstyle{definition}

\theoremstyle{definition}

\newcommand{\card}[1]{\vert #1 \vert}

\newcommand{\psf}{\mathcal{P}}
\newcommand{\gbox}{[\forall]}
\newcommand{\gdi}{[\exists]}

\newcommand{\pos}{\mathtt{p}}

\author{
  Sebastian Enqvist \and  Fatemeh Seifan \and Yde Venema
}
\title{Expressiveness of the modal $\mu$-calculus on monotone neighborhood structures}

\begin{document}
\maketitle

\begin{abstract}
We characterize the expressive power of the modal $\mu$-calculus on monotone neighborhood structures, in the style of the Janin-Walukiewicz theorem for the standard modal $\mu$-calculus. For this purpose we consider a monadic second-order logic for monotone neighborhood structures. Our main result shows that the monotone modal $\mu$-calculus corresponds exactly to the fragment of this second-order language that is invariant for neighborhood bisimulations. 
\end{abstract}

\section{Introduction}

The \textit{modal $\mu$-calculus} was introduced, in its present form, by D. Kozen in \cite{koze:resu83}. It functions as a general specification language for labelled transition systems, encompassing many systems used in formal verification of processes, including propositional dynamic logic ($\mathtt{PDL}$) and many temporal logics, like computation-tree logic ($\mathtt{CTL}$). In fact, \textit{any} logic for labelled transition systems that is \textit{invariant for bisimulation}, and that can be translated into monadic second order logic, can be seen as a fragment of the $\mu$-calculus. This is due to the Janin-Walukiewicz theorem \cite{jani:expr96}, which states that the modal $\mu$-calculus captures \textit{exactly} the bisimulation invariant fragment of monadic second-order logic. This result is the counterpart for the $\mu$-calculus of  van Benthem's characterization theorem for basic modal logic \cite{bent:moda76}, which isolates modal logic as the bisimulation invariant fragment of first-order logic. 

In this paper, we consider the modal $\mu$-calculus for \textit{monotone neighborhood structures}, rather than labelled transition systems. 
Monotone neighborhood structures are a  generalization of Kripke frames that are used to give semantics for modal logics that do not satisfy the distribution law for conjunctions: 
$$\Box(\varphi \wedge \psi) \leftrightarrow (\Box \varphi \wedge \Box \psi)$$
With monotone neighborhood semantics, this equivalence is weakened to the implication from left to right. One good example of an application of this is \textit{alternating-time temporal logic}, which is useful to reason about state-based evolving systems consisting of several interacting processes. In particular, it allows reasoning about conditions that can be ``forced'' by one process regardless of how other parts of the system behave. It can happen that condition $\varphi$ can be forced, as well as condition $\psi$, but not the conjunction of both. Hence, the distribution law for boxes over conjunctions should not hold.

 In this context, fixed point operators are a natural addition to the basic modal logic; for example, the formula
$$\mu p. \Box p \vee \Box \psi$$
informally expresses the ``liveness property'' that it is possible to force the condition $\psi$ to hold using some  finite sequence of actions. To see why, note that with the intended semantics for the box, the formula $\Box {\perp}$ is equivalent to ${\perp}$ (it should never be possible to force a contradiction to hold by any action!). So by using  ordinal approximations of the least fixpoint, we see that the formula is equivalent to an infinite disjunction:
$$ \Box \psi \vee \Box\Box \psi \vee \Box \Box\Box \psi \vee...$$

 The basic modal logic of monotone neighborhood structures is known as \textit{monotone modal logic}, and following this nomenclature we shall refer to the $\mu$-calculus variant of this logic as the \textit{monotone $\mu$-calculus}. Several results on monotone modal logic have been obtained in a MSc thesis by Hansen \cite{hans:mono03}, including Sahlqvist correspondence and completeness, a Goldblatt-Thomason theorem and Craig interpolation. The interpolation theorem was later strengthened to a uniform interpolation result by Santocanale and Venema \cite{sant:unif10}. Several results on the monotone $\mu$-calculus are also known, including uniform interpolation \cite{seif:unif14} and decidability in exponential time \cite{cirs:expt09}.

In this paper, we characterize the expressive power of the $\mu$-calculus on monotone neighborhood structures precisely, by exhibiting it as the bisimulation invariant fragment of a suitable monadic second-order logic, in the style of the Janin-Walukiewicz theorem. This result suggests that the role of the modal $\mu$-calculus as a ``universal'' specification language extends beyond Kripke frames and labelled transition systems, to the more general setting of monotone neighborhood structures. The monadic second-order language $\mathtt{NMSO}$ that we use as the ``yardstick'' language here was introduced, in a more general setting, in \cite{enqv:mona15}. There, it was shown that the fragment of $\mathtt{NMSO}$ that is invariant for \textit{global} bisimulations corresponds to an extension of the monotone $\mu$-calculus with the \textit{global} modalities. Our main result shows that a formula of this latter system is invariant for bisimulation if, and only if, it is equivalent to a formula without any occurrence of the global modalities. From this, together with the characterization result from \cite{enqv:mona15}, we get our Janin-Walukiewicz theorem for the monotone $\mu$-calculus.
 
\section{Technical preliminaries}

In this section we introduce the rudimentary technical concepts that will be used throughout the paper: neighborhood structures, neighborhood bisimulations and the monotone $\mu$-calculus.

\subsection{Neighborhood structures}

We start by introducing the basic concept of a neighborhood structure, or more specifically, neighborhood \textit{frames} and \textit{models}. These structures provide the standard semantics for monotone modal logic.

\begin{defi}
A \textit{neighborhood frame} is a pair $(S,\sigma)$ where $S$ is a set, and $\sigma$ is a map from $S$ to $\psf \psf (S)$. Elements of $\sigma(s)$, for $s \in S$, are called \textit{neighborhoods} of $s$. A neighborhood frame $(S,\sigma)$ is said to be \textit{monotone} if, for all $s \in S$ and all $Z,Z' \subseteq S$: if $Z \in \sigma(s)$ and $Z \subseteq Z'$, then $Z' \in \sigma(s)$ too.
\end{defi}

Note that Kripke frames can be seen as special instances of monotone neighborhood frames, with the extra condition that the neighborhoods of any point are closed under arbitrary intersections. Given any such neighborhood frame $(S,\sigma)$, we can define the Kripke frame $(S,R_\sigma)$ by setting
$$(u,v) \in R_\sigma \text{ iff } v \in \bigcap \sigma (u)$$
Conversely, given any Kripke frame $(S,R)$ we can define the neighborhood frame $(S,\sigma_R)$ by setting
$$\sigma_R(u) = \{Z \subseteq S \mid  \forall v: \text{ if } (u,v)\in R \text{ then } v \in Z\}$$
From now on, we shall fix a countably infinite set of propositional variables $\mathit{Var}$. We shall also refer to these as \textit{second-order variables}.  
\begin{defi}
A \textit{monotone neighborhood model}, or just a neighborhood model, is a triple $(S,\sigma,V)$ where $(S,\sigma)$ is a monotone neighborhood frame, and $V : \mathit{Var} \to \psf(S)$ is a \textit{valuation} for the propositional variables. Given a neighborhood model $\mathbb{S}$, and given $s \in S$, the pair $(\mathbb{S},s)$ will be called a \textit{pointed} neighborhood model. 
\end{defi}
The fundamental concept in model theory of modal logic is that of a \textit{bisimulation}. We assume that the reader is familiar with the concept of a bisimulation between Kripke models. By now, monotone neighborhood models are also equipped with a fairly standard notion of bisimulation:
\begin{defi}
Let $\mathbb{S}$ and $\mathbb{S}'$ be any pair of monotone neighborhood models. A relation $R \subseteq S \times S'$ is said to be a \textit{neighborhood bisimulation} if, for all $s \in S$ and all $s' \in S'$ with $(s,s') \in R$, the following clauses hold:
\begin{enumerate}
\item For all $Z \in \sigma(s)$ there is some $Z' \in \sigma'(s')$ such that, for all $u' \in Z'$, there is some $u \in Z$ with $(u,u')\in R$
\item For all $Z' \in \sigma'(s')$ there is some $Z \in \sigma(s)$ such that, for all $u \in Z$, there is some $u' \in Z'$ with $(u,u') \in R$
\end{enumerate}
 \end{defi}
The pointed models $(\mathbb{S},s)$ and $(\mathbb{S}',s')$ are said to be \textit{neighborhood bisimilar} if there is a bisimulation $R$ with $(s,s') \in R$. We denote this situation by $(\mathbb{S},s) \sim (\mathbb{S}',s')$. 

The following observation is standard.
\begin{prop}
Neighborhood bisimulations are closed under unions: if $\{R_i\}_{i \in I}$ is a family of neighborhood bisimulations between $\mathbb{S}$ and $\mathbb{S}'$, then  $\bigcup_{i \in I}R_i$ is a neighborhood bisimulation too. 
\end{prop}
We shall also need the following variation of the concept of neighborhood bisimulations later:
\begin{defi}
A neighborhood bisimulation $R$ between neighborhood models $\mathbb{S}$ and $\mathbb{S}'$ is said to be \textit{global} if it is full on both $S$ and $S'$. In other words, it satisfies the following zig-zag conditions:
\begin{enumerate}
\item For every $s \in S$ there is some $s' \in S'$ such that $(s,s') \in R$
\item For every $s' \in S'$ there is some $s \in S$ such that $(s,s' ) \in R$
\end{enumerate}
\end{defi}
The pointed models $(\mathbb{S},s)$ and $(\mathbb{S}',s')$ are said to be \textit{globally neighborhood bisimilar} if there is a global bisimulation $R$ with $(s,s') \in R$. We denote this situation by $(\mathbb{S},s) \sim_g (\mathbb{S}',s')$.

A useful model construction that should be familiar from standard modal logic is that of \textit{disjoint union}, or \textit{co-product}. Given an indexed family $\{\mathbb{S}_i\}_{i \in I}$ of neighborhood models, where $\mathbb{S}_i = (S_i,\sigma_i,V_i)$ consider the disjoint union $\coprod_{i \in I } S_i$ and let 
$$\iota_j : S_j \to \coprod_{i \in I} S_i$$ be the insertion of $S_i$ into this disjoint union. We can supply this set with a neighborhood map $\tau$ by setting, for all $j \in I$, all $u \in S_j$ and all $Z \subseteq \coprod_{i \in I} S_i$:
$$Z \in \tau(\iota_i(u)) \text{ iff } \iota_i^{-1}[Z] \in \sigma_i(u)$$
Furthermore, we define a valuation $W$ over the disjoint union by setting
$$\iota_i(u) \in W(p) \text{ iff } u \in V_i(p)$$
From now on, we shall not take care to distinguish between  $u$ and $\iota_i(u)$. We define the neighborhood model
$$\coprod_{i \in I} \mathbb{S}_i = (\coprod_{i \in I} S_i, \tau,W)$$
and call this the disjoint union of the $\mathbb{S}_i$. Given two models $\mathbb{S}$ and $\mathbb{S}'$, we denote their disjoint union simply by $\mathbb{S} + \mathbb{S}'$.

\begin{prop}
For each $j \in J$, the graph of the insertion map $\iota_j$ is a neighborhood bisimulation between $\mathbb{S}_j$ and $\coprod_{i \in I}\mathbb{S}_i$. Hence, for all $u \in S_j$, we have
$$(\mathbb{S}_j,u) \sim (\coprod_{i \in I} \mathbb{S}_i, u)$$
\end{prop}

\subsection{The monotone modal  $\mu$-calculus}

In this section we present the monotone modal $\mu$-calculus, with and without the global modality. First, the language $\mu \mathtt{NML}$ is defined by the following grammar:
$$ \varphi ::= p \mid \neg p \mid \varphi \wedge \varphi \mid \varphi \vee \varphi \mid \Box \varphi \mid \Diamond \varphi \mid \mu p. \varphi \mid \nu p. \varphi$$
where $p$ ranges over $\mathit{Var}$, and in the formula $\eta p . \varphi$ for $\eta \in \{\mu, \nu\}$, the variable $p$ does not appear under the scope of a negation.
Note that we have presented the language in negation normal form here, so that negations only appear in front of propositional variables. Alternatively we could have presented the language with an unrestricted use of negations and proved a negation normal form theorem, but since this is entirely standard by now we skip this little extra step. 

The extended language $\mu \mathtt{NML}_g$, with the global modalities, is presented by the following grammar: 
$$ \varphi ::= p \mid \neg p \mid \varphi \wedge \varphi \mid \varphi \vee \varphi \mid \Box \varphi \mid \Diamond \varphi \mid  [\forall] \varphi \mid [\exists] \varphi \mid \mu p. \varphi \mid \nu p. \varphi$$
Free and bound variables of a formula are defined as usual. 

Given a valuation $V : \mathit{Var} \to \psf(S)$, a variable $p$ and a subset $Z \subseteq S$, the valuation $V[p \mapsto Z]$ is defined to be the unique valuation that is like $V$, except that it sends $p$ to $Z$. Given a model $\mathbb{S} = (S,\sigma, V)$ and $s \in S$, interpretations of formulas in $\mu \mathtt{NML}_g$ are defined as follows:

\begin{enumerate}
\item $\llbracket p \rrbracket_\mathbb{S} = V(p)$ 
\item $\llbracket \neg p \rrbracket_\mathbb{S} = S \setminus V(p)$
\item $\llbracket \varphi \wedge \psi \rrbracket_\mathbb{S} = \llbracket \varphi \rrbracket_\mathbb{S} \cap \llbracket \psi \rrbracket_\mathbb{S} $
\item $\llbracket \varphi \vee \psi \rrbracket_\mathbb{S} =  \llbracket \varphi \rrbracket_\mathbb{S} \cap \llbracket \psi \rrbracket_\mathbb{S}$
\item $\llbracket \Box \varphi \rrbracket_\mathbb{S} = \{u \in S \mid \llbracket \varphi \rrbracket_\mathbb{S} \in \sigma(u)\}$
\item $\llbracket \Diamond \varphi \rrbracket_\mathbb{S} = \{u \in S \mid \llbracket (S \setminus \varphi \rrbracket_\mathbb{S}) \notin \sigma(u)\}$
\item $\llbracket [\forall] \varphi \rrbracket_\mathbb{S} = \{u \in S \mid \llbracket \varphi \rrbracket_\mathbb{S} = S\}$
\item $\llbracket [\exists] \varphi \rrbracket_\mathbb{S} = \{u \in S \mid \llbracket \varphi \rrbracket_\mathbb{S} \neq \emptyset \}$
\item $\llbracket \mu p. \varphi \rrbracket_\mathbb{S} = \bigcap \{Z \mid  \llbracket \varphi \rrbracket_{(S,\sigma,V[p \mapsto Z])} \subseteq Z\}  $
\item $\llbracket \nu p. \varphi \rrbracket_\mathbb{S} = \bigcup \{Z \mid Z \subseteq \llbracket \varphi \rrbracket_{(S,\sigma,V[p \mapsto Z])} \}  $
\end{enumerate}
We write $(\mathbb{S},s)\vDash \varphi$ for $s \in \llbracket \varphi \rrbracket_\mathbb{S}$.
\begin{defi}
A formula $\varphi$ is said to be \textit{well-named} if:
\begin{itemize}
\item No variable appears both bound and free in $\varphi$, and 
\item For every bound variable $p$ of $\varphi$, there is exactly one subformula of $\varphi$ of the form $\eta p. \psi$ for $\eta \in \{\mu,\nu\}$.
\end{itemize}
The formula  $ \psi$ is then called the \textit{binding definition} of $p$ in $\varphi$, and is denoted by $\mathcal{D}(p,\varphi)$.
\end{defi}
From now on we shall assume that all formulas are well-named, since it is easy to show that every formula is equivalent to a well-named one. Given a pair of bound variables $p,q$ of $\varphi$, we say that  $p$ \textit{ranks higher} than $q$ if $p $ appears free in the binding definition of $q$.  If $\varphi$ has a subformula of the form $\mu p. \mathcal{D}(p,\varphi)$, then we say that $p$ is a \textit{$\mu$-variable}.

\subsection{Game semantics}

The \textit{evaluation game} for a (well-named) formula $\varphi$ in $\mu \mathtt{NML}_g$ and $\mu \mathtt{NML}$ relative to a neighborhood model $\mathbb{S} = (S,\sigma,V)$, denoted $\mathcal{G}(\mathbb{S},\varphi)$, is a two-player game played between ``$\exists$'', or ``Eloise'', and ``$\forall$'', or ``Abelard''. Intuitively, Eloise tries to show that the formula is true at some point in the model, while Abelard tries to refute this same claim. 

The game board has two types of positions. First the ``basic positions'', the set of which is defined to be:
$$S \times \mathsf{Sub}(\varphi)$$
consisting of pairs $(s,\psi)$ with $s \in S$ and $\psi$ any subformula of $\varphi$.
Second, the ``intermediate positions'', which are the elements of the set:
$$\{\forall,\exists\} \times \psf(S) \times   \mathsf{Sub}(\varphi)$$
consisting of triples $(\mathsf{P},Z,\varphi)$ where $\mathsf{P}$ is either $\exists$ or $\forall$, $Z$ is a subset of $S$ and $\varphi$ is a formula. 

We assign a player and a set of admissible moves to a given position as described in the table below. Here, we recall that $\mathcal{D}(p,\varphi)$ denotes the binding definition of the bound variable $p$ in $\varphi$.

\begin{table}[H]
    \centering
\begin{tabular}{|l|c|l|}
\hline
Position  & Player  &  Admissible moves 
\\ \hline
      $(s,\psi \vee \theta)$ 
   & $\exists$  
   & $\{ (s,\psi),(s,\theta)\}$
\\

\hline
      $(s,\psi \wedge \theta)$ 
   & $\forall$  
   & $\{ (s,\psi),(s,\theta)\}$
                                                 \\
\hline
      $(s,x)$ 
   & --  
   & $\{ (s,\mathcal{D}(x,\varphi))\}$ \\
\hline
      $(s,\Box \psi)$ 
   & $\exists$  
   & $\{ (\forall, Z,\psi) \mid Z \in \sigma(s)\}$\\
\hline
      $(s,\Diamond \psi)$ 
   & $\forall$  
   & $\{ (\exists,Z,\psi) \mid Z \in \sigma(s)\}$\\
\hline
      $(\forall,Z,\psi)$ 
   & $\forall$  
   & $\{(t,\psi) \mid t \in Z\}$ \\
\hline
      $(\exists,Z,\psi)$ 
   & $\exists$  
   & $\{(t,\psi) \mid t \in Z\}$\\
\hline
      $(s,[\forall]\psi)$ 
   & $\forall$  
   & $\{(t,\psi) \mid t \in S\}$\\
\hline
      $(s,[\exists]\psi)$ 
   & $\exists$  
   & $\{(t,\psi) \mid t \in S\}$\\
\hline
      $(s,p) $ with $s \in V(p)$ 
   & $\forall$  
   & $\emptyset$\\
\hline
      $(s,p) $ with $s \notin V(p)$ 
   & $\exists$  
   & $\emptyset$\\
\hline
      $(s,\neg p) $ with $s \in V(p)$ 
   & $\exists$  
   & $\emptyset$\\
\hline
      $(s,\neg p) $ with $s \notin V(p)$ 
   & $\forall$  
   & $\emptyset$\\
\hline
      $(s,\top) $
   & $\forall$  
   & $\emptyset$\\
\hline
      $(s,{\perp}) $ 
   & $\exists$  
   & $\emptyset$\\
\hline
  \end{tabular}
\end{table}

The concepts of a \textit{match}, a \textit{partial match} and a \textit{strategy} are defined as usual. A finite match is lost by the player who got stuck, and an infinite match is won by $\forall$ if the unique highest ranking variable that appears infinitely many times on the match is a $\mu$-variable. Otherwise the winner is $\exists$. A strategy $\chi$ is \textit{winning} for player $\mathsf{P}$ at position $\pos$ if $\mathsf{P}$ wins every $\chi$-guided match starting at $\pos$, i.e. every match starting at $\pos$ and in which all moves by $\mathsf{P}$ are made according to $\chi$. Given a player $\mathsf{P} \in \{\exists,\forall\}$, the set of positions of $\mathcal{G}(\mathbb{S},\varphi)$ at which $\mathsf{P}$ has a winning strategy are denoted by
$$ \mathit{Win}_\mathsf{P}(\mathcal{G}(\mathbb{S},\varphi))$$
We now list four important results about the game semantics. These can all be proved by entirely routine methods, so we omit the arguments.
\begin{prop}[Adequacy of Game Semantics]
For any neighborhood model $\mathbb{S}$, any $s \in S$, and any formula $\varphi \in \mu \mathtt{NML}_g$, we have
$$(\mathbb{S},s)\vDash \varphi \text{ iff } (s,\varphi) \in \mathit{Win}_\exists(\mathcal{G}(\mathbb{S},\varphi)) $$
\end{prop}

\begin{prop}[History-free Determinacy]
For any formula $\varphi$, any neighborhood model $\mathbb{S}$ and any position $\pos$ in $\mathcal{G}(\mathbb{S},\varphi)$, we have
$$\pos  \in \mathit{Win}_\exists(\mathcal{G}(\mathbb{S},\varphi)) \cup \mathit{Win}_\forall(\mathcal{G}(\mathbb{S},\varphi))$$
Furthermore, if $\mathsf{P} \in \{\forall,\exists\}$ has a winning strategy that is winning at the position $\pos$, then that player has a winning strategy $\chi$ at $\pos$ which is \textit{positional}. This means that for all partial matches $\pi$ and $\pi'$ starting at $\pos$, such that the last position of both these partial matches is the same, we have $\chi(\pi) = \chi(\pi')$. 
\end{prop}

\begin{prop}[Bisimulation Invariance]
\label{invariancelocal}
Let $\varphi$ be any formula in $\mu\mathtt{NML}$, let $\mathbb{S}$ and $\mathbb{S}'$ be  neighborhood models that are related by some neighborhood bisimulation $R$, and let $\mathsf{P} \in \{\exists,\forall\}$. Then, for every subformula $\psi$ of $\varphi$ and any pair of states $s \in S$ and $s' \in S'$ such that $(s,s') \in R$, we have:
$$(s,\psi) \in \mathit{Win}_{\mathsf{P}}(\mathcal{G}(\mathbb{S},\varphi)) \text{ iff } (s',\psi) \in \mathit{Win}_{\mathsf{P}}(\mathcal{G}(\mathbb{S}',\varphi))$$
\end{prop}

\begin{prop}[Global Bisimulation Invariance]
\label{invarianceglobal}
Let $\varphi$ be any formula in $\mu\mathtt{NML}_g$, let $\mathbb{S}$ and $\mathbb{S}'$ be  neighborhood models that are related by some \textit{global} neighborhood bisimulation $R$, and let $\mathsf{P} \in \{\exists,\forall\}$. Then, for every subformula $\psi$ of $\varphi$ and any pair of states $s \in S$ and $s' \in S'$ such that $(s,s') \in R$, we have:
$$(s,\psi) \in \mathit{Win}_{\mathsf{P}}(\mathcal{G}(\mathbb{S},\varphi)) \text{ iff } (s',\psi) \in \mathit{Win}_{\mathsf{P}}(\mathcal{G}(\mathbb{S}',\varphi))$$
\end{prop}

The game semantics for $\mu \mathtt{NML}$ and $\mu \mathtt{NML}_g$ will be the key technical tools that we use to obtain our Janin-Walukiewicz theorem for $\mu\mathtt{NML}$.

\section{Expressive completeness of $\mu\mathtt{NML}$}

\subsection{The monadic second-order logic of monotone neighborhood structures}

We now present a monadic second-order language for monotone neighborhood structures. This language is very closely related to the monadic second-order logic introduced by Walukiewicz in \cite{jani:expr96}. Following both the presentation in \cite{jani:expr96} and \cite{enqv:mona15}, we shall use a ``single-sorted'' presentation of monadic second-order logic here, without the presence of any individual (first-order) variables. This is no restriction, since individual variables can be ``simulated'' by monadic second-order variables. This is due to the simple fact that, given a second-order variable $p$, there is a second-order formula $\mathtt{Sing}(p)$ stating that the value of $p$ is a singleton set.

The syntax of the monadic second-order language $\mathtt{NMSO}$ is given by the following grammar:

$$\varphi ::= sr(p) \mid p \subseteq q \mid \Box(p,q) \mid \varphi \vee \varphi \mid \varphi \wedge \varphi \mid \neg \varphi \mid \exists p. \varphi $$

Semantics relative to a pointed model with $\mathbb{S} = (S,\sigma,V)$ are defined as follows:
\begin{enumerate}
\item $(\mathbb{S},s)\vDash sr(p)$ iff $V(p) = \{s\}$
\item $(\mathbb{S},s)\vDash p \subseteq q$ iff $V(p) \subseteq V(q)$
\item $(\mathbb{S},s) \vDash \Box(p,q)$ iff $V(q) \in \sigma(t) $ for all $t \in V(p)$
\item Standard clauses for Boolean connectives
\item $(\mathbb{S},s)\vDash \exists p. \varphi $ iff, for some $Z \subseteq S$ we have
 $$(S,\sigma,V[p \mapsto Z],s)\vDash \varphi$$ 
\end{enumerate}
Note that there is a hidden quantifier pattern encoded in an atomic formula $\Box(p,q)$ of the form ``$\forall \exists \forall$'': for all states $t$ in the extension of $p$ there exists a neighborhood $Z$ of $t$ such that, for all members $t'$ of $Z$, $t'$ satisfies $q$. 

A formula $\varphi$ of $\mathtt{NMSO}$ is said to be \textit{invariant for neighborhood bisimulations}, or just \textit{bisimulation invariant}, if whenever $(\mathbb{S},s) \sim (\mathbb{S}',s')$ we have
$$(\mathbb{S},s)\vDash \varphi \text{ iff } (\mathbb{S}',s')\vDash \varphi $$
Invariance for global neighborhood bisimulations is defined in the same way. We denote by $\mathtt{NMSO}{/}{\sim}$ the fragment of $\mathtt{NMSO}$ that is invariant for neighborhood bisimulations, and similarly  $\mathtt{NMSO}{/}{\sim_g}$ denotes the fragment of $\mathtt{NMSO}$ that is invariant for global neighborhood bisimulations. We can then state the main result about $\mathtt{NMSO}$ from \cite{enqv:mona15} as follows:
\begin{theo}
\label{globalchar}
$$\mathtt{NMSO}{/}{\sim_g} \equiv \mu \mathtt{NML}_g$$
\end{theo}
Here, the equivalence symbol $\equiv$ is intended to have the meaning that, for every formula $\varphi$ of $\mathtt{NMSO}{/}{\sim_g}$ there is a formula $\varphi' $ of $\mu \mathtt{NML}_g$ true in exactly the same pointed models as $\varphi$, and vice versa. In words: $\mu \mathtt{NML}_g$ is the fragment of $\mathtt{NMSO}$ that is invariant for global bisimulations.

Our main contribution here is to strengthen this result, and show that $\mu \mathtt{NML}$ is the bisimulation invariant fragment of $\mathtt{NMSO}$:

\begin{theo}
\label{goal}
$$\mathtt{NMSO}{/}{\sim} \equiv \mu \mathtt{NML}$$
\end{theo}

We prove only one part of this inclusion here, leaving the difficult direction for later. Given any formula of $\mu \mathtt{NML}$, we shall find an equivalent formula of $\mathtt{NMSO}$. More precisely, for every formula $\varphi$ of $\mu \mathtt{NML}$, and any $p \in \mathit{Var}$, we shall construct a formula $\mathtt{Eq}(\varphi,p)$ such that
$$(\mathbb{S},s)\vDash \mathtt{Eq}(\varphi,p) \text{ iff } V(p) = \llbracket \varphi \rrbracket_\mathbb{S}$$
From this, we can obtain our translation  $c$ of $\mu \mathtt{NML}$ into $\mathtt{NMSO}$ as follows: given any formula $\varphi$, let $p$ be a fresh variable that does not appear in $\varphi$ and let $q$ be any variable that does not appear in $\mathtt{Eq}(\varphi,p)$. Then we set
$$c(\varphi) := \exists p \exists q ( sr(q) \wedge q \subseteq p \wedge \mathtt{Eq}(\varphi,p))$$
Then clearly $\varphi$ is equivalent to $c(\varphi)$.

The reader can easily construct the formulas $\mathtt{Eq}(p,q)$ and $\mathtt{Eq}(\neg p, q)$, so we leave out the details. The steps for conjunction and disjunction are also fairly simple; the main observation needed for all these cases is that the basic set theoretic operations like union, intersection and complement are definable in $\mathtt{NMSO}$. For box- and diamond-formulas, we proceed as follows:
$$\mathtt{Eq}(\Box\varphi ,p) := \forall q(q \subseteq p \leftrightarrow \exists r (\mathtt{Eq}(\varphi,r) \wedge \Box (q,r)))$$
where $q,r$ are fresh variables that do not appear in $\varphi$.  Since diamond-formulas are dual to box-formulas, we leave this simple case out. 

Finally, we have to take care of the fixpoint formulas. We treat only the case for the least fixpoint formulas, since the case for greatest fixpoints is dual to this one. We set:
\begin{displaymath}
\begin{array}{lcl}
 \mathtt{Eq}(\mu q. \varphi, p) & := & \mathtt{Eq}(\varphi[p/q],p)  \\
& \wedge & \forall p' (\mathtt{Eq}(\varphi[p'/q],p') \rightarrow p \subseteq p' )
\end{array}
\end{displaymath}
Here, $r$ and $p'$ are fresh variables that do not appear in $\mu q.\varphi$.
This formula simply says that the value of $p$ is a least fixpoint of the monotone function on $\psf(S)$ determined by the formula $\varphi$, which by the Knaster-Tarski fixpoint theorem ensures that the formula $\mathtt{Eq}(\mu q.\varphi,p)$ has the right meaning. Here,  $\varphi[p/q]$ denotes the result of uniformly substituting $p$ for $q$ in $\varphi$.

\subsection{The monotone modal $\mu$-calculus inside $\mu \mathtt{NML}_g$}

We shall now prove Theorem \ref{goal}, and to do this we shall characterize $\mu\mathtt{NML}$ inside $\mu \mathtt{NML}_g$ in order to derive the main characterization result from Theorem \ref{globalchar}. First, we prove a simple little lemma in ZFC set theory:
\begin{lemma}
The language $\mu \mathtt{NML}_g$ has a Löwenheim-Skolem number. In other words,
there exists a cardinal $\kappa$ such that every satisfiable formula in $\mu \mathtt{NML}_g$ is satisfiable in a pointed model $(\mathbb{S},s)$ with $\card{S} \leq \kappa$.
\end{lemma}

\begin{proof}
For each satisfiable formula $\varphi$, let $\lambda(\varphi)$ be the smallest cardinal number such that $\varphi$ is satisfiable in a pointed model of cardinality at most $\lambda(\varphi)$. By the Axiom Schema of Replacement, the class $\{\lambda(\varphi) \mid \varphi \in \mu \mathtt{NML}_g\}$ forms a \textit{set} (a countable set, in fact). Hence, there is a cardinal $\kappa$ that is greater than each $\lambda(\varphi)$, and the proof is done. 
\end{proof}

Let $\kappa$  be the smallest Löwenheim-Skolem number for $\mu \mathtt{NML}_g$, so that $\kappa$ has the property described in the previous lemma. From now on, we assume that we have at our disposal a fixed neighborhood model $ \mathbb{U} = (U,\gamma,V)$ such that, for every  pointed model $(\mathbb{S},s)$ with $\card{S} \leq \kappa$, there is some $u \in U$ such that $(\mathbb{S},s)\sim (\mathbb{U},u)$. It is not hard to see that such a model does exist: just take a disjoint union of all neighborhood models defined on subsets of some fixed set of cardinality $\kappa$. Since the collection of all these models forms a set, the disjoint union is well defined.  

\begin{lemma}
\label{usim}
For every model $\mathbb{S}$ with $\card{S} \leq \kappa$, there is a global neighborhood bisimulation $R$ between $\mathbb{U} + \mathbb{S}$ and $\mathbb{U}$ such that $(u,u) \in R$ for each $u \in U$. Hence, for all $u \in U$, we have:
$$  (\mathbb{U},u) \sim_g (\mathbb{U} + \mathbb{S},u)  $$
\end{lemma}
\begin{proof}
For every $t \in S$ there is neighborhood bisimulation $R_t$ between $\mathbb{S}$ and $\mathbb{U}$ such that $(t,t') \in R_t$ for some $t' \in U$.  This $R_s$ is a neighborhood bisimulation between $\mathbb{U} + \mathbb{S}$  and $\mathbb{U}$ as well. Furthermore, the identity relation $\mathsf{Id}_U$ on $U$ is a neighborhood bisimulation between $\mathbb{U}$ and $\mathbb{U} + \mathbb{S}$. Hence, since neighborhood bisimulations are closed under unions, we get that
$$\mathsf{Id}_U \cup \bigcup_{t \in S} R_t$$
is a neighborhood bisimulation. Since this relation $R \subseteq (U + S) \times U$ is full on both $U + S$ and $U$, and since $(u,u) \in R$ for each $u \in U$, the result follows.  
\end{proof}

Now, given any fixed formula $\varphi$ of $\mu \mathtt{NML}_g$, we define a translation $t : \mathsf{Sub}(\varphi) \to \mu \mathtt{NML}$ inductively as follows:
\begin{itemize}
\item $t(\Box \psi) = \Box t(\psi)$ and $t(\Diamond\psi) = \Diamond t(\psi)$ 
\item $t(\gbox \psi) = {\perp}$ if for all $u \in U $ we have $(u,\psi) \in \mathit{Win}_\exists(\mathcal{G}(\mathbb{U},\varphi))$
\item Otherwise, if there is some $u \in U$ such that  $(u,\psi) \in \mathit{Win}_\exists(\mathcal{G}(\mathbb{U},\varphi))$, set $t([\forall]\psi) = \top $
\item  $t([\exists] \psi) = \top$ if there is some $u \in U $ for which we have $(u,\psi) \in \mathit{Win}_\exists(\mathcal{G}(\mathbb{U},\varphi))$
\item Otherwise, if for every $u \in U$ we have  $(u,\psi) \in \mathit{Win}_\forall (\mathcal{G}(\mathbb{U},\varphi))$, set $t(\gdi\psi) = {\perp}$
\item $t(p) = p$ and $t(\neg p) = \neg p$
\item $t (\mu p. \psi) = \mu p. t(\psi)$ and similarly for $\nu$.
\end{itemize}

\begin{lemma}
\label{binddef}
Suppose $t$ is the translation associated with a well-named formula $\varphi$, and let $p$ be a bound variable of $\varphi$ that appears in $t(\varphi)$ also. Then $t(\varphi)$ is a well-named formula too, and we have $$t(\mathcal{D}(p,\varphi)) = \mathcal{D}(p,t(\varphi))$$
\end{lemma}

\begin{proof}
It is easy to show that $t(\varphi)$ is well-named. So suppose that $p$ is a bound variable in $\varphi$ that also appears in $t(\varphi)$. Since $\varphi$ is well-named, every occurrence of $p$ in $\varphi$ is in the form of a subformula of $\mathcal{D}(p,\varphi)$. Hence, clearly, the subformula $\eta p. \mathcal{D}(p,\varphi)$ of $\varphi$ (where $\eta \in \{\mu,\nu\}$) cannot be in the scope of any occurrence of $[\exists]$ or $[\forall]$. This means that $t(\eta p. \mathcal{D}(p,\varphi)) = \eta p. t( \mathcal{D}(p,\varphi)) $ is a subformula of $t(\varphi)$, and it follows that $t(\mathcal{D}(p,\varphi))$. Since $t(\varphi)$ is well-named we get 
$\mathcal{D}(p,t(\varphi)) = t(\mathcal{D}(p,\varphi))$
as required.
\end{proof}

\begin{lemma}
\label{mainlemma}
For every pointed model $(\mathbb{S},s)$ with $\card{S} \leq \kappa$, and for every formula $\varphi$ in $\mu \mathtt{NML}_g$, we have
$$(\mathbb{S},s) \vDash t(\varphi) \text{ iff } (\mathbb{U} + \mathbb{S},s) \vDash \varphi$$
where $t : \mathsf{Sub}(\varphi) \to \mu \mathtt{NML}$ is the translation associated with the formula $\varphi$.
\end{lemma}
\begin{proof}
First, for each basic position $(v,\psi)$ with $v \in U + S$ and such that $(v,\psi) \in \mathit{Win}_\exists(\mathcal{G}(\mathbb{U}+ \mathbb{S},\varphi))$, pick a strategy $\tau_{(v,\psi)} $ that is winning at $(v,\psi)$ in $\mathcal{G}(\mathbb{U}+ \mathbb{S},\varphi)$. Note that for every position $(u,\psi) \in \mathit{Win}_\exists(\mathcal{G}(\mathbb{U},\varphi))$, we have  $(u,\psi) \in \mathit{Win}_\exists(\mathcal{G}(\mathbb{U}+ \mathbb{S},\varphi))$ also by Theorem \ref{invarianceglobal} and Lemma \ref{usim}. Given a position $\pos$ in $\mathcal{G}(\mathbb{S} + \mathbb{U},\varphi)$ let $t(\pos)$ denote the pair $(v,t(\psi))$ if $\pos$ is $(v,\psi)$, and let $t(\pos)$ be $(\mathsf{P},Z,t(\psi))$ if $\pos$ is $(\mathsf{P},Z,\psi)$ for $\mathsf{P} \in \{\exists,\forall\}$.

Now, suppose $\exists$ has a winning strategy in  $\mathcal{G}(\mathbb{S},t(\varphi))$ at position $(s,t(\varphi))$. Then by Theorem \ref{invariancelocal}, $\exists$ has a winning strategy $\chi$ in $\mathcal{G}(\mathbb{U} + \mathbb{S},t(\varphi))$ at position $(s,t(\varphi))$ too, since $t(\varphi)$ is a formula in $\mu \mathtt{NML}$ and since
$$ (\mathbb{S},s) \sim (\mathbb{U} + \mathbb{S},s) $$
We are going to construct a winning strategy $\chi^* $ for $\exists$ in $\mathcal{G}(\mathbb{U} + \mathbb{S},\varphi)$ at the starting position $(s,\varphi)$. We shall define $\chi^*$ by induction on the length of a partial $\pi$, and show by simultaneous induction on the length of a $\chi^*$-guided match $\pi$  that one of the following two cases holds:
\begin{description}
\item[Case 1:] $\pi$ is of the form 
$(\pos_0,...,\pos_k)$ where $(t(\pos_0),...,t(\pos_k))$ is a $\chi$-guided partial match, or:
\item[Case 2:] there is some position $(v,\psi)$ that appears on $\pi$ such that 
$$(v,\psi) \in \mathit{Win}_\exists(\mathcal{G}(\mathbb{U}+ \mathbb{S},\varphi))$$
Furthermore, if $(v,\psi)$ is the \textit{first} position in $\pi$ for which this holds, then the final segment of $\pi$ starting with the first occurrence of $(v,\psi)$ is a $\tau_{(v,\psi)}$-guided partial match.
\end{description}
It will follow that every infinite $\chi^*$-guided match $M$ starting at $(s,\varphi)$ either corresponds to a $\chi$-guided infinite match $t[M]$ starting at $(s,t(\varphi))$, such that  $M$ is a loss for $\exists$ if and only if $t[M]$ is, or (apart from some finite initial segment) $M$ is a $\tau_{(u,\psi)}$-guided match starting at $(u,\psi)$ for some pair $(u,\psi)$ with $u \in U$. Hence, since $\chi$ and all the $\tau_{(u,\psi)}$ are winning strategies, $\exists$ wins every infinite $\chi^*$-guided match. 

Clearly the induction hypothesis holds for the match consisting only of the initial position $(s,\varphi)$. So suppose that the induction hypothesis holds for a match $\pi$ of length $k$. If the last position of $\pi$ belongs to $\exists$ then we show how to define the strategy $\chi^*$ on $\pi$ in such a way that the induction hypothesis remains true for $\chi^*(\pi)$, and if the last position of $\pi$ belongs to $\forall$ then we show that the induction hypothesis is true for each  partial match resulting from a possible move by $\forall$. If $\pi$ falls under Case 2 then the argument is trivial since then we just follow some strategy $\tau_{(u,\psi)}$ that was picked at the first occurence of a position $(u,\psi)$ with $u \in U$. So we consider Case 1, and divide it into a number of subcases depending on the shape of the last position on $\pi$. We shall assume here, without loss of generality, that the strategy $\chi$ was positional. We only treat the non-trivial cases, leaving the others to the reader:

Suppose the last position of $\pi$ is $(v,p)$ where $p$ is a bound variable. Then  $t(v,p) = (v,t(p)) = (v,p)$, and $p$ must appear as a bound variable of $t(\varphi)$. The only extension of $\pi$ is with the position $(v, \mathcal{D}(p,\varphi))$, and by Lemma \ref{binddef} we have
$$t(v, \mathcal{D}(p,\varphi)) = (v, t(\mathcal{D}(p,\varphi))) = (v,\mathcal{D}(p,t(\varphi)))$$
which shows that $t[\pi]\cdot t(v, \mathcal{D}(p,\varphi))$ is a $\chi$-guided match, as required.

Finally, we treat the case where the last position of $\pi$ is of the form $(v,\gbox \psi)$ or $(v,\gdi\psi)$. Suppose the first of these two cases holds. Then this position belongs to $\forall$, and we must show that the inductive hypothesis holds for each extension of $\pi$ given by a choice made by $\forall$. Now, for every $u \in U$ we must have $(u,\psi) \in \mathit{Win}_\exists (\mathcal{G}(\mathbb{U},\varphi))$, for otherwise we would have $t(\gbox  \psi) = {\perp}$, which means that $t(v,\gbox \psi) = (v,{\perp})$, an immediate loss for $\exists$. Hence, we have $(w,\psi) \in \mathit{Win}_\exists (\mathcal{G}(\mathbb{U}+ \mathbb{S},\varphi))$ for all $w \in U + S$ as well, since by Lemma \ref{usim} every pointed model $(\mathbb{U} + \mathbb{S},w)$ is globally bisimilar with some pointed model $(\mathbb{U},u)$. This means that for every choice $(w,\psi)$ by $\forall$, the strategy $\tau_{(w,\psi)}$ is defined, and so the induction hypothesis remains true.

Dually, if the last position of $\pi$ is of the form $(v,\gdi \psi)$, then there must be some $u \in U$ such that $(u,\psi) \in \mathit{Win}_\exists(\mathcal{G}(\mathbb{U},\varphi))$, since otherwise $t(v,\gdi \psi)$ is $(v,{\perp})$ and $\pi$ is a loss for $\exists$. Hence $(u,\psi) \in \mathit{Win}_\exists(\mathcal{G}(\mathbb{U} + \mathbb{S},\varphi))$ as well, and so the strategy $\tau_{(u,\psi)}$ is defined. So if we set $\chi^*(\pi) = (u,\psi)$ then the induction hypothesis remains true. We have now defined the strategy $\chi^*$ so that $\exists$ never gets stuck, and so that she wins every infinite $\chi^*$-guided match.

Conversely, suppose that $\forall$ has a winning strategy in $\mathcal{G}(\mathbb{S},t(\varphi))$ at the start position $(s,t(\varphi))$. Then we can prove, using an argument that is completely symmetric with the one we used above, that there is a winning strategy for $\forall$ in $\mathcal{G}(\mathbb{U} + \mathbb{S},\varphi)$ at the position $(s,\varphi)$. Hence, the proof is done.
\end{proof}
We can now prove the main technical result of this paper:

\begin{theo}
\label{main}
$$ \mu \mathtt{NML}_g{/}{\sim} \equiv \mu \mathtt{NML}$$
\end{theo}

\begin{proof}
Suppose a formula $\varphi$ of $\mu\mathtt{NML}_g$ is invariant for neighborhood bisimulations, but not equivalent to any formula of $\mu \mathtt{NML}$. Then, in particular, $\varphi$ is not equivalent to $t(\varphi)$. So there are two possible cases:
\begin{description}
\item[Case 1:] $\varphi \wedge \neg t(\varphi)$ is satisfiable
\item[Case 2:] $\neg \varphi \wedge t(\varphi)$ is satisfiable
\end{description}
Here, we are using the fact that $\mu \mathtt{NML}_g$ is closed under negation, even though we have presented the formulas in negation normal form. 

So suppose Case 1 holds. Then by our choice of $\kappa$ there is a  pointed model $(\mathbb{S},s)$ such that $\card{S} \leq \kappa$, and such that $$(\mathbb{S},s)\vDash \varphi \wedge \neg t(\varphi)$$
This is a contradiction, since Lemma \ref{mainlemma} gives:
\begin{displaymath}
\begin{array}{lcl}
 (\mathbb{S},s)\vDash \varphi & \Leftrightarrow & (\mathbb{U} + \mathbb{S},s) \vDash \varphi \\
& \Leftrightarrow & (\mathbb{S},s)\vDash t(\varphi)
\end{array}
\end{displaymath}
Case 2 is handled in the same manner.
\end{proof}

Finally, we can conclude that the monotone modal $\mu$-calculus is indeed the neighborhood bisimulation invariant fragment of monadic second-order logic over neighborhood structures:
\begin{proof}[Proof of Theorem \ref{goal}]
If a formula $\varphi$ of $\mathtt{NMSO}$ is invariant for all neighborhood bisimulations, then it is invariant for global neighborhood bisimulations in particular. Hence, it is equivalent to a formula $\varphi'$ in $\mu \mathtt{NML}_g$ by Theorem \ref{globalchar}. By Theorem \ref{main}, it immediately follows that $\varphi'$ is equivalent to a formula of $\mu \mathtt{NML}$, and hence so is $\varphi$.
\end{proof}

\section{Concluding remarks}

Our main result showed that the Janin-Walukiewicz theorem for the modal $\mu$-calculus remains true for the $\mu$-calculus and monadic second-order logic interpreted on monotone neighborhood structures. This resolves an open problem in \cite{enqv:mona15}, and provides an expressive completeness result for the monotone $\mu$-calculus. 

An interesting question is whether the full language $\mathtt{NMSO}$ can be characterized by some fixpoint logic for neighborhood structures, in the style of \cite{walu:mona96} where it is shown that monadic second-order logic on trees is equivalent to a first-order fixpoint logic. Some of the ground work has already been done here: the main tool used for the characterization in \cite{walu:mona96} is a translation of monadic second-order logic into \textit{parity automata} over trees. A similar translation was achieved in \cite{enqv:mona15}, so it is possible that this result can be used for a translation of $\mathtt{NMSO}$ into a suitable first-order fixpoint language, relative to ``tree-like'' neighborhood structures. We leave this as a problem for future research.  

\bibliography{mu,automata,nabla,extra}
\bibliographystyle{plain}
\end{document}